\documentclass[10pt, a4paper, journal]{IEEEtran}
\usepackage[noadjust]{cite}
\usepackage{amsmath,amssymb,amsfonts}
\usepackage{authblk}
\usepackage{amsthm,bbm}

\newtheorem{proposition}{Proposition}[]

\newtheorem{remark}{Remark}[]

\newtheorem{definition}{Definition}[]
\newtheorem{problem}{Problem}[]
\usepackage{lineno}
\usepackage{lipsum}
\usepackage{graphicx}
\usepackage{textcomp}
\usepackage{color,soul}
\usepackage{subfig}

\usepackage[linesnumbered,ruled,vlined]{algorithm2e}
\usepackage{algorithmicx}
\SetKwInput{KwInput}{Input}
\SetKwInput{KwOutput}{Output}

\usepackage{xcolor}
\usepackage{scalerel}
\usepackage{float}
\usepackage{tikz}
\usetikzlibrary{svg.path}

\usepackage{tikz,xcolor,hyperref}

\definecolor{lime}{HTML}{A6CE39}
\DeclareRobustCommand{\orcidicon}{%
	\begin{tikzpicture}
	\draw[lime, fill=lime] (0,0) 
	circle [radius=0.16] 
	node[white] {{\fontfamily{qag}\selectfont \tiny ID}};
	\draw[white, fill=white] (-0.0625,0.095) 
	circle [radius=0.007];
	\end{tikzpicture}
	\hspace{-2mm}
}

\foreach \x in {A, ..., Z}{%
	\expandafter\xdef\csname orcid\x\endcsname{\noexpand\href{https://orcid.org/\csname orcidauthor\x\endcsname}{\noexpand\orcidicon}}
}

\def\BibTeX{{\rm B\kern-.05em{\sc i\kern-.025em b}\kern-.08em
    T\kern-.1667em\lower.7ex\hbox{E}\kern-.125emX}}
\begin{document}

\title{\LARGE{Network Graph Generation through Adaptive Clustering\\ and Infection Dynamics: A Step Towards Global Connectivity}}


%




\author[$*,\#$]{Aniq Ur Rahman}
\author[$*,\#$]{Fares Fourati}
\author[$\dagger$]{Khac-Hoang Ngo}
\author[$+$]{Anish Jindal}
\author[$*$]{Mohamed-Slim Alouini}

\affil[$\#$]{ {\small{equal technical contribution}} \vspace{0.15cm}}
\affil[$*$]{King Abdullah University of Science and Technology, Thuwal, Kingdom of Saudi Arabia.}
\affil[$\dagger$]{Chalmers Univerity of Technology, Gothenburg, Sweden.}
\affil[$+$]{Univesity of Essex, Colchester, United Kingdom. \vspace{-0.85cm}}

\maketitle

\begin{abstract}
More than 40\% of the world's population is not connected to the internet, majorly due to the lack of adequate infrastructure.
Our work aims to bridge this digital divide by proposing solutions for network deployment in remote areas. Specifically, a number of access points (APs) are deployed as an interface between the users and backhaul nodes (BNs). The main challenges include designing the number and location of the APs, and connecting them to the BNs. In order to address these challenges, we first propose a metric called \emph{connectivity ratio} to assess the quality of the deployment. Next, we propose an agile search algorithm to determine the number of APs that maximizes this metric and perform clustering to find the optimal locations of the APs.
Furthermore, we propose a novel algorithm inspired by infection dynamics to connect all the deployed APs to the existing BNs economically. To support the existing terrestrial BNs, we investigate the deployment of non-terrestrial BNs, which further improves the network performance in terms of average hop count, traffic distribution, and backhaul length. Finally, we use real datasets from a remote village to test our solution.
\end{abstract}
\vspace{5pt}
\vspace{-0.5cm}

\section{Introduction}
\IEEEPARstart{T}{he} research and development of future communication networks has been driven towards providing faster and more reliable connection for urban and developed regions. Provisioning connectivity to remote regions has been relegated to the bottom. In 2019, about 87\% of people in developed countries were connected to the Internet, while in striking contrast only 19\% of people in the least developed countries were connected~\cite{SustainableDevelopment2021}. This means that the most vulnerable to the COVID-19 pandemic were also those do not have access to online tools to respond to the impact of the pandemic. The pandemic has thus exacerbated the lingering digital divide. This calls for a consensus to provide broadband connectivity to rural/remote regions in 6G~\cite{chaoub20216g,dang2021big}.

One of the main challenges in establishing broadband connectivity in remote areas is
the deployment of mobile backhaul
solutions. Due to high deployment costs, network operators have been reluctant to deploy fiber optics. Therefore, rural/remote backhaul relies mostly on wireless solutions, such as microwave, free-space optics (FSO), and satellite~\cite{yaacoub2020efficient}. In any case, taking both capital expenditures and long-term operational expenditures into account, only few of backhaul nodes~(BNs) would be deployed in denser areas. However, a non-negligible fraction of rural population are scattered in isolated villages with geographic barriers, such as mountains and forests, to the main BNs. Therefore, the deployment of access points~(APs) in proximity to the users for fronthaul connectivity should be carefully designed. A cost analysis of different solutions for fronthaul and backhaul connectivity in rural areas was reported in~\cite{yaacoub2020efficient}. Design and analysis of rural networks based on different solutions have been reported, such as long-range Wi-Fi~\cite{Hamid2011self}, drones~\cite{Maurilio2021coverage}, and satellites~\cite{ogutu2021techno}. 
For example, Viasat has a fleet of satellites capable of providing global coverage in the Ka-band \cite{viasat2021}. Having such satellite BNs can bring connectivity to the most remote locations on earth, thereby bridging the digital divide.
The existing literature does not provide a general algorithm to deploy frugal networks \cite{khaturia2020connecting} in any location, so as to connect its unconnected population to the internet.
Therefore, in this work, we address the problem of network deployment in rural/remote areas in a systematic manner. Given a set of few terrestrial BNs available in a sparsely populated region, we aim to design the deployment of APs to serve the scattered users. Specifically, we optimize the number and locations of the  APs and the network configuration to connect those APs to the BNs. We also explore the use of non-terrestrial BNs to further improve the connectivity performance of the network.

\begin{figure*}
    \centering
    \includegraphics[width=1.7\columnwidth]{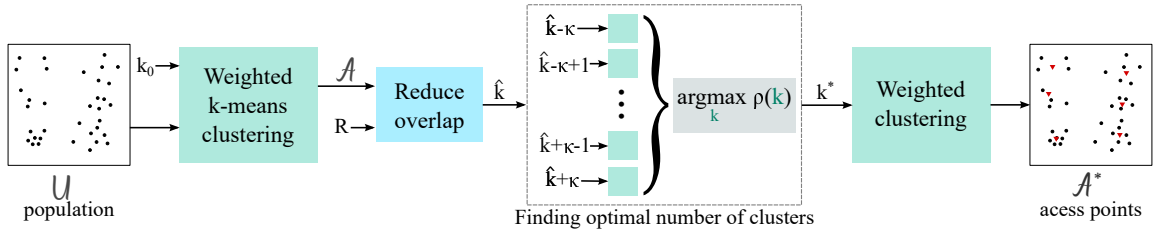}
    \caption{Generating AP locations through iterative clustering, based on the spatial distribution of population.}
    \label{fig:cluster}
\end{figure*}

\textbf{Contributions and Organization.}
The contributions of our work are summarized as follows.
\begin{itemize}
    \item We propose a metric called {\em connectivity ratio} to assess the quality of network deployment, balancing user coverage and deployment cost. This metric is used as the optimization objective to determine the optimal set of APs.
    \item We convert the connectivity ratio maximization problem to relaxed sub-problems, where we first determine the number of APs through an agile search algorithm and then perform weighted clustering to place the APs.
    \item We propose a novel algorithm inspired by infection dynamics, to economically connect all the deployed APs to the BNs.
    \item We investigate the effect of adding  non-terrestrial BNs on the network performance.
\end{itemize}

The remainder of the paper is organised as follows. In Sec.~II, we define the coverage ratio and solve the maximization of this metric. In Sec.~III, we propose a graph generation algorithm inspired by infection dynamics, and also investigate the use of non-terrestrial BNs. The paper is finally concluded in Sec.~IV with some comments on future works.

\textbf{Notation.}
We denote the set of integers from $m$ to $n$ by $[\![m,n]\!]$; ${\rm Area}(\mathcal{U})$ denotes the area of the convex hull of a set $\mathcal{U}$ in a 2D space; $\|\cdot\|$ denotes the Euclidean distance.

%

\section{Access Point Deployment} \label{sec:acces_point_deployement}
We consider a set of users $\mathcal{U}$ scattered in a two-dimensional region of interest in the presence of a small number of terrestrial BNs. To cover the users, we deploy a set of APs $\mathcal{A}$. Each AP $\mathbf{a}$ in $\mathcal{A}$ covers the users in a circular region of radius $R$, denoted by $u_{\mathbf{a}} \triangleq \mathcal{U} \cap \mathcal{B}\left( \mathbf{a}, R \right)$, where $\mathcal{B}(\mathbf{x}, l)$ denotes a circle of radius $l$ centred at $\mathbf{x}$. The number of users covered by at least one AP in $\mathcal{A}$ is denoted by  $\mathcal{C}_{\mathcal{A}} \triangleq \bigcup_{\mathbf{a}\in \mathcal{A}} u_\mathbf{a}$. Note that due to the limited range of the APs, not all the users are guaranteed to be covered, i.e., $|\mathcal{C}_\mathcal{A}| \le |\mathcal{U}|$.

To effectively deploy the APs, we need to determine the number of APs and their positions. On the one hand, the number of APs needs to be sufficiently large so that $\mathcal{A}$ can collectively cover the region. On the other hand, an excessive number of APs increases the deployment cost. This calls for a design metric that balances between user coverage and deployment cost, which remains unclear in the literature. To this end, we propose a metric called \textit{connectivity ratio}. 


\begin{definition}
The connectivity ratio $\rho(\mathcal{A})$ associated with the set of APs $\mathcal{A}$ is defined as:
\begin{align} \label{eq:rho}
    \rho(\mathcal{A}) \triangleq \frac{\left| \mathcal{C}_{\mathcal{A}} \right|^2}{|\mathcal{A}||\mathcal{U}|}.
\end{align}
\end{definition}
The connectivity ratio is the product of two important metrics: (i) the average number of users per AP  $\frac{|\mathcal{C}_{\mathcal{A}}|}{|\mathcal{A}|}$, and (ii) the coverage ratio $\frac{|\mathcal{C}_{\mathcal{A}}|}{|\mathcal{U}|}$. A deployment with large connectivity provides connectivity to a majority of the users while minimizing the number of APs, as interpreted in the following remark. 
\begin{remark}
Consider two deployments $\mathcal{A}_1$ and $\mathcal{A}_2$ with $\rho(\mathcal{A}_1) > \rho(\mathcal{A}_2)$. If the two deployments use the same number of APs, i.e., $|\mathcal{A}_1| = |\mathcal{A}_2|$, then $\mathcal{A}_1$ covers a larger number of users, i.e., $|\mathcal{C}_{\mathcal{A}_1}| > |\mathcal{C}_{\mathcal{A}_2}|$. If they covers the same number of users, i.e., $|\mathcal{C}_{\mathcal{A}_1}| = |\mathcal{C}_{\mathcal{A}_2}|$, then $\mathcal{A}_1$ uses a smaller number of APs, i.e., $|\mathcal{A}_1| < |\mathcal{A}_2|$, thus saves the deployment cost.
\end{remark}
Therefore, to balance between maximizing coverage and minimizing the deployment cost, we maximize the connectivity ratio.

\begin{problem}\label{prob:ratio}
Generate a set of AP locations $\mathcal{A}^*$ such that the connectivity ratio $\rho(\mathcal{A}^*)$ is maximized, i.e.,
\begin{align}
    \mathcal{A}^* = \arg \max_{\mathcal{A}} \rho(\mathcal{A})
    \label{eq:p1}
\end{align}
\end{problem}
For a fixed number of APs $k = |\mathcal{A}|$, one can optimize the positions of the APs by clustering \cite{inaba1994applications,arthur2006slow,scikit-learn} the set of all users $\mathcal{U}$ into $k$ clusters and place an AP at the centroid of each cluster. We denote the set of APs generated from this clustering by $\mathcal{A} = \psi(k)$. Nevertheless, in our setup, $k$ is unknown \emph{a priori} and also needs to be optimized. 
In order to simplify Problem~\ref{prob:ratio} while exploiting existing clustering algorithms, we decouple the optimization of $k$ and of the positions of the APs as follows. First, we optimize the number of clusters $k$ as
\begin{align} \label{eq:opt_k}
    k^* = \arg \max_{k \in [\![ 1, |\mathcal{U}|]\!]} \rho(\psi(k)).
\end{align}
Then, the set of clusters is generated as $\mathcal{A}^* = \psi(k^*)$. 

\begin{remark}
Solving the optimization of $k$ in~\eqref{eq:opt_k} through exhaustive search has a worst-case complexity \cite{inaba1994applications} of: 
$\mathcal{O}\left( \sum_{k=1}^{|\mathcal{U}|} k |\mathcal{U}|^{k + 1} \right),$ which is computationally expensive, especially when the number of users $|\mathcal{U}|$ is large. Therefore, we aim to reduce the search space.
\end{remark}
With a slight abuse of notation, hereafter we write $\rho(\psi(k))$ simply as $\rho(k)$ for convenience. When $k$ increases from a small value, the connectivity ratio $\rho(k)$ increases since for small number of clusters, each added AP helps covering more users. Specifically, in this regime, $|\mathcal{C}_{\mathcal{A}}|^2$ increases faster than $k$, thus it follows from~\eqref{eq:rho} that $\rho(k)$ increases. However, for large values of $k$, the coverage zones of the APs begin to overlap and cover the same population. Once the majority of the users have been connected, adding more APs increases the denominator of $\rho(k)$ while the numerator $|\mathcal{C}_{\mathcal{A}}|^2$ remains approximately the same. This suggests that $\rho(k)$ decreases after a certain value of $k$. This is made precise in the following proposition, where we invoke the definition of covering in Appendix~\ref{app:covering_packing}.

\begin{proposition} \label{prop:kmax}
Let $k_{\max}$ be the $R$-covering number of $\mathcal{U}$, i.e., $k_{\max} = N(R,\mathcal{U})$ (see Definition~\ref{def:covering}). Then $\rho(k)$ is a decreasing function of $k$ for $k \geq k_{\max}$. Furthermore, $k_{\max}$ is bounded as
\begin{align} \label{eq:bound_kmax}
        \frac{{\rm Area}(\mathcal{U})}{\pi R^2} \le k_{\max} \le \frac{4{\rm Area}(\mathcal{U}_{R/2})}{\pi R^2},
\end{align}
where $\mathcal{U}_{R/2}$ denotes the union of the circles of radius $R/2$, each centered at a point in $\mathcal{U}$.
\end{proposition}
\begin{proof}
By definition of covering, $k_{\max}$ is the smallest number of APs for which all users are covered, i.e., $|\mathcal{C}_{\psi(k_{\max})}| = |\mathcal{U}|$. Since $|\mathcal{C}_{\psi(k)}|$ is non-decreasing in $k$, it holds that $|\mathcal{C}_{\psi(k)}| = |\mathcal{U}|$, $\forall k > k_{\max}$.
Note that $k_{\max}$ is guaranteed to be finite as $k_{\max} \leq |\mathcal{U}|$. It follows that for $k>k_{\max}$, the connectivity ratio is given by $\rho(k) = \frac{|\mathcal{U}|}{k}$, which is obviously a decreasing function of $k$. The bound \eqref{eq:bound_kmax} follows directly from Proposition~\ref{prop:covering_number}.
\end{proof}

It follows from Proposition~\ref{prop:kmax} that $\rho(k) \le \rho(k_{\max})$, $\forall k \in [\![ k_{\max},|\mathcal{U}|]\!]$. Therefore, the search space in~\eqref{eq:opt_k} can be reduced without loss of optimality to $[\![1,k_{\max}]\!]$ i.e.,  
\begin{align}
    k^* = \arg \max_{k \in [\![ 1, |\mathcal{U}|]\!]} \rho(k) 
    = \arg \max_{k \in [\![1, k_{\max}]\!]}   \rho(k) .
\end{align}
This $k^*$ is guaranteed to exist as the search space is discrete and bounded.

Although the search space has been reduced, it remains big for large $k_{\max}$. Specifically, we see from~\eqref{eq:bound_kmax} that $k_{\max}$ is large when $R$ is small and when ${\rm Area}(\mathcal{U})$ is large, i.e., the set of users is scattered in a large region, which is the case in remote/rural areas. To further reduce the space, we propose a heuristic method to estimate a value $\hat{k}$ in proximity to the optimal value $k^*$ and then search in the neighborhood of $\hat{k}$. Specifically, it follows from Proposition~\ref{prop:covering_number} that $k_{\max}$ is lower-bounded by $P(2R,\mathcal{U})$, which is the largest number of APs such that the circles with radius $R$ centered at these APs do not overlap. Since $\rho(k)$ starts decreasing when the overlap between the clusters becomes significant, we predict that $\rho(k)$ is maximized near $P(2R,\mathcal{U})$, i.e., $k^* \in [\![P(2R,\mathcal{U}) - \kappa, P(2R,\mathcal{U}) + \kappa]\!]$ for sufficiently large $\kappa$. Therefore, we first estimate $P(2R,\mathcal{U})$ and then search for $k^*$ in the neighborhood of the estimate. To estimate $P(2R,\mathcal{U})$, we start from a value $k_0$ larger than $P(2R,\mathcal{U})$ (e.g., $k_0 = \frac{4{\rm Area}(\mathcal{U}_{R/2})}{\pi R^2}$), partition the population into $k_0$ clusters, and then progressively remove the APs whose radius-$R$ circle intersects with other APs' circles. In this way, we expect to form a dense $2R$-packing of $\mathcal{U}$ and thus the resulting number of APs $\hat{k}$ closely approaches the packing number $P(2R,\mathcal{U})$. Then, we find $k^*$ using an exhaustive search the extensively reduced search space $[\![\hat{k} - \kappa, \hat{k} + \kappa]\!]$. Finally, we perform clustering with $k^*$ clusters to determine the positions of the APs. The proposed method is presented in Algorithm~\ref{algo:aploc} and illustrated in Fig.~\ref{fig:cluster} and Fig.~\ref{fig:ksearch}.


\begin{figure}[h!]
    \centering
    \includegraphics[width=0.9\columnwidth]{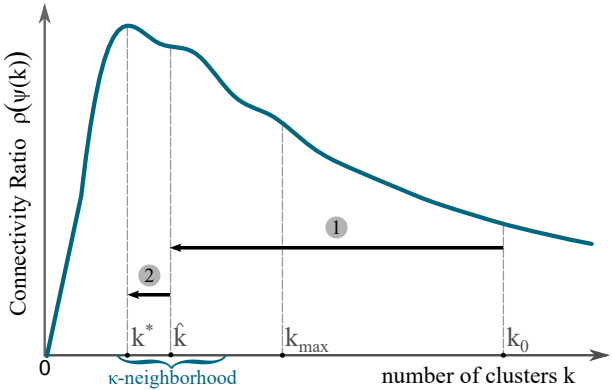}
    \caption{Illustration of the variation of the connectivity ratio $\rho(k)$ and the search of $k^*$.}
    \label{fig:ksearch}
\end{figure}

\begin{algorithm} 
\DontPrintSemicolon
\caption{Learning APs Positions}
\label{algo:aploc}
\KwInput{the population $\mathcal{U}$, radius $R$, and initial guess $k_{0}$}
\KwOutput{set of APs $\mathcal{A}^*$}  
$\mathcal{A} \gets \texttt{clustering}(k_0, \mathcal{U})$  \Comment{Weighted clustering}\;
$N \gets 0$ \Comment{Count of overlapping APs}\;
$\mathcal{A}' = \mathcal{A}$\; 
\For{$\mathbf{a}$ in $\mathcal{A}$}
{
    $\mathcal{A}' = \mathcal{A}' \setminus \{ \mathbf{a} \}$ \;
        \If{$\exists~ \mathbf{a}' \in \mathcal{A}'$ such that $\lVert \mathbf{a} - \mathbf{a}' \rVert \leq 2R $}
        {
            $N \gets N+1$ \Comment{APs overlap} \;
        }
}
$ \hat{k} \gets k_0 - N+1$ \Comment{Refining k}\;
$k^* \gets \arg \displaystyle\max_{ k \in [\![\hat{k} - \kappa, \hat{k} + \kappa]\!] } \rho(k) $ \Comment{Exhaustive search}\;
$\mathcal{A}^* \gets \texttt{clustering}(k^*, \mathcal{U})$ \Comment{Weighted clustering}\;
\end{algorithm}

We next demonstrate our algorithm using a real dataset of the population of Kilimambogo, Kenya~\cite{facebook2021}, which has one of the lowest gross domestic product (GDP) per capita in the world.
The connectivity ratio for various coverage radius $R$ is shown in Fig.~\ref{fig:trend3}. 
Moreover, in Fig.~\ref{fig:map1}, we show the optimal location of the APs and the spatial distribution of the population for $R=750$~m in an area of roughly 400 km$^2$.

\begin{figure}[h!]
    \centering
    \includegraphics[width=\columnwidth]{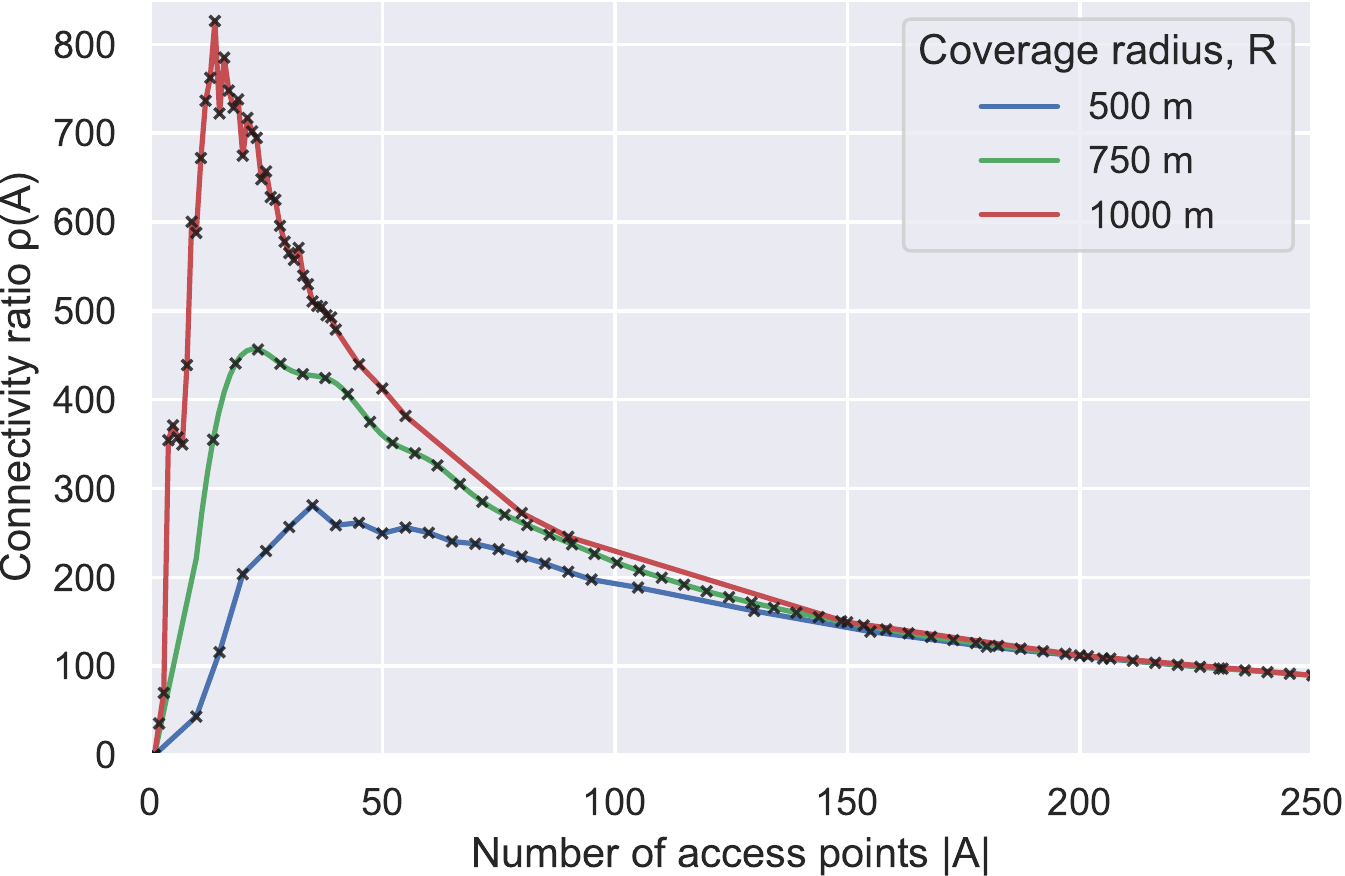}
    \caption{Connectivity ratio $\rho(\mathcal{A})$ as a function of the number of APs for Kilimambogo, Kenya.}
    \label{fig:trend3}
\end{figure}

\begin{figure}[h!]
    \centering
    \includegraphics[width=0.9\columnwidth]{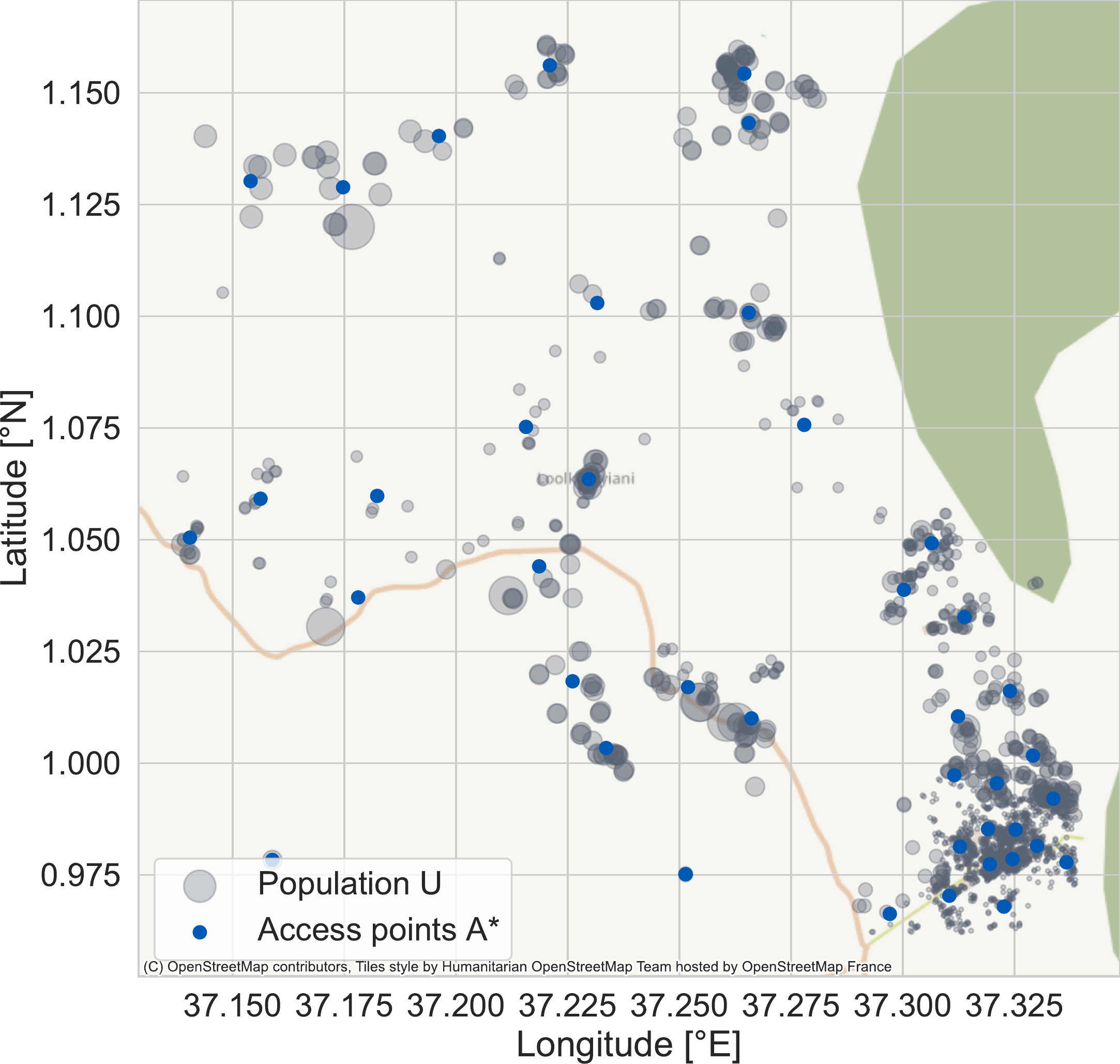}
    \caption{Optimal AP locations that maximizes $\rho(\mathcal{A})$, for a region in Kilimambogo, Kenya ($R=750$~m) for $|\mathcal{A^*}|=39$.}
    \label{fig:map1}
\end{figure}

\begin{figure*}[t!]
    \centering
    \includegraphics[width=1.8\columnwidth]{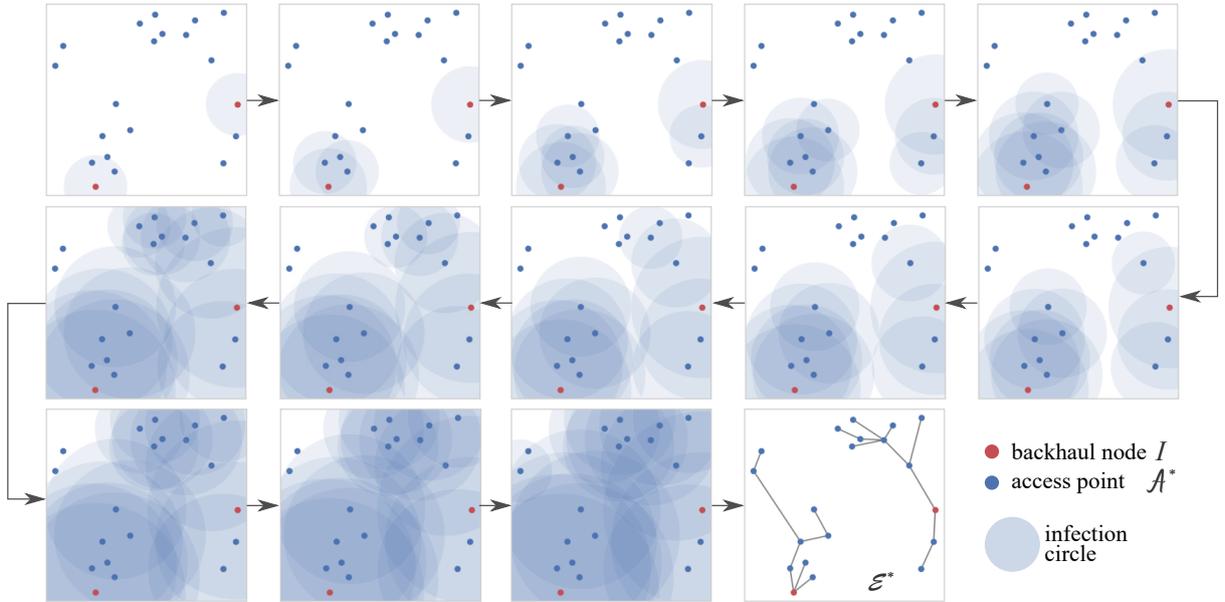}
    \caption{{Infection algorithm in action:} generating the network topology.}
    \label{fig:zombie_sample}
\end{figure*}

\newpage
\section{Backhaul Link Generation} \label{sec:network_generation}
Through clustering, we have obtained the set of APs $\mathcal{A}^*$, and we also have available with us, the set of terrestrial BNs $\mathcal{I}$.\footnote{The location of the existing BNs, i.e., cell towers is obtained through Open Cell ID~\cite{opencellid2021}.} Now, our focus is to provide backhaul to all the APs while optimally using the backhaul resources. We want every AP to be connected to one of the BNs, such that the total length of the backhaul resource we utilize is minimized. Formally, we frame the problem as follows.
\begin{problem} 
Generate a graph $\mathfrak{G}(\mathcal{A}^* \cup \mathcal{I}, \mathcal{E}^*)$ such that all the points in $\mathcal{A}^*$ are directly or indirectly connected to one of the points in $\mathcal{I}$ and the sum of the edges' length is minimized:
\begin{align}
    \mathcal{E}^* &= \arg \min_{\mathcal{E}} \sum_{ \mathbf{e} \in \mathcal{E}} \lVert \mathbf{e} \rVert \nonumber \\
    &\text{such that} \quad \forall \mathbf{a} \in \mathcal{A}^*, \exists  \mathbf{a}' \in \mathcal{I} : \mathbf{a} \leftrightarrow \mathbf{a}',
    \label{eq:p2}
\end{align}
where $\mathbf{a} \leftrightarrow \mathbf{a}'$ denotes that the vertices $\mathbf{a}$ and $\mathbf{a}'$ are directly or indirectly connected, and $\lVert \mathbf{e} \rVert$ is the length of the edge $\mathbf{e}$.
\label{prob:graph}
\end{problem}

\begin{remark}
The exact solution of Problem~\ref{prob:graph} obtained through combinatorics has a complexity of $\mathcal{O}\left( 2^{ \binom{|\mathcal{A}^* \cup \mathcal{I}|}{2}}  \right)$, which makes it computationally infeasible for large networks.
\end{remark}
To reduce complexity, we propose in the next subsection an algorithm to approximately solve Problem~\ref{prob:graph}.

\subsection{Infection Algorithm}
Our proposed algorithm is inspired by the concept of \textit{infection dynamics} \cite{liu2003propagation, may2001infection}. The vertices belonging to $\mathcal{I}$ (infected) compete among themselves to infect the vertices in $\mathcal{A}^*$ by sweeping out a circle whose radius increases non-linearly with time. 
Once the circle touches a vertex in $\mathcal{A}^*$, this vertex also gets infected and begins competing with the other infected vertices. This process is illustrated in Fig.~\ref{fig:zombie_sample}. 
Mathematically, we model the growth of the radius of an infected vertex $\mathbf{i}$ at time $t$ as
\begin{align}
    \dot r_{\mathbf{i}}(t) = \left ( \alpha + \frac{\beta}{1 + \gamma \, r_{\mathbf{i}}^2(t)} \right) \mathbbm{1}\{t > t_0^{\mathbf{i}}\},
\end{align}
where $\alpha, \beta$ are hyper parameters, $t_0^{\mathbf{i}}$ denotes the time of infection of vertex $\mathbf{i}$, and $\mathbbm{1}\{\cdot\}$ denotes the indicator function.
Instead of running the algorithm over continuous time, we discretize it into time-steps of size $\delta$ each. In this way, at each time-step $j$, we update the radius $r_{\mathbf{i}}[j]$ and speed $s_{\mathbf{i}}[j]$ of the infected vertex  $\mathbf{i}$ as
\begin{align}
    r_{\mathbf{i}}[j] &= \left( r_{\mathbf{i}}[j-1] + \delta \cdot s_{\mathbf{i}}[j-1]\right)\mathbbm{1}\{j > j_0^{\mathbf{i}}\},
    \label{eq:radius}
\end{align}
\begin{align}
    s_{\mathbf{i}}[j] &= \left( \alpha + \frac{\beta}{1 + \gamma \, r_{\mathbf{i}}^2[j]} \right)\mathbbm{1}\{j > j_0^{\mathbf{i}}\},
    \label{eq:speed}
\end{align}
where $j_0^{\mathbf{i}} = \Big\lceil \frac{t_0^{\mathbf{i}}}{\delta} \Big\rceil$. Initially $j_0^{\mathbf{i}} = 0, \forall \mathbf{i} \in \mathcal{I}$. The vertex connects to its infector, in other words, it connects to the vertex whose circle touches it first. The algorithm terminates once all the vertices are infected and the resulting graph gives the network topology. 
\begin{figure}[h!]%
 \centering
 \subfloat[Speed of infection.
 ]{\includegraphics[width=0.47\linewidth]{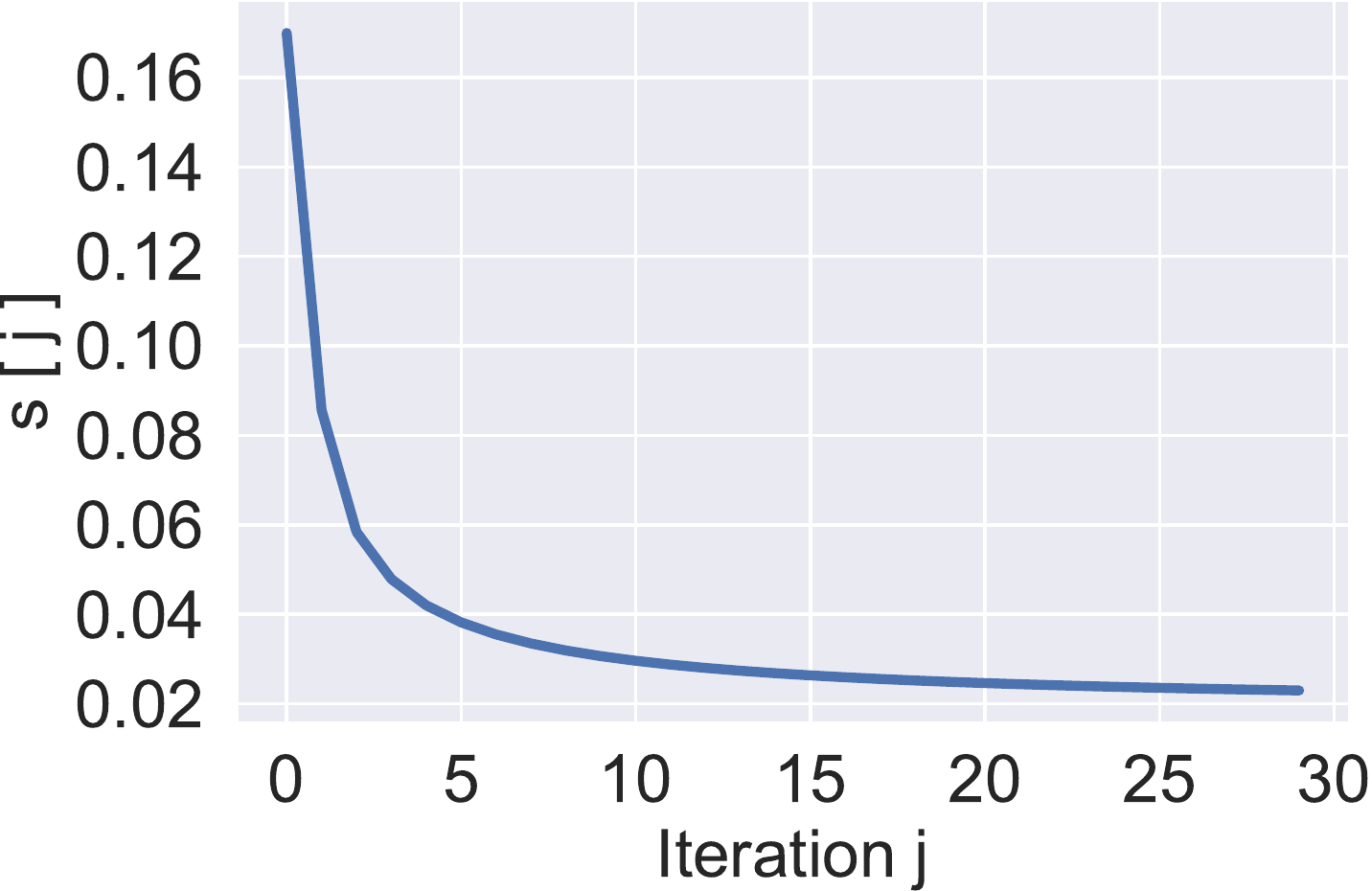}\label{fig:inf_speed}} \hspace{1pt}
 \subfloat[Radius of infection circle.
 ]{\includegraphics[width=0.47\linewidth]{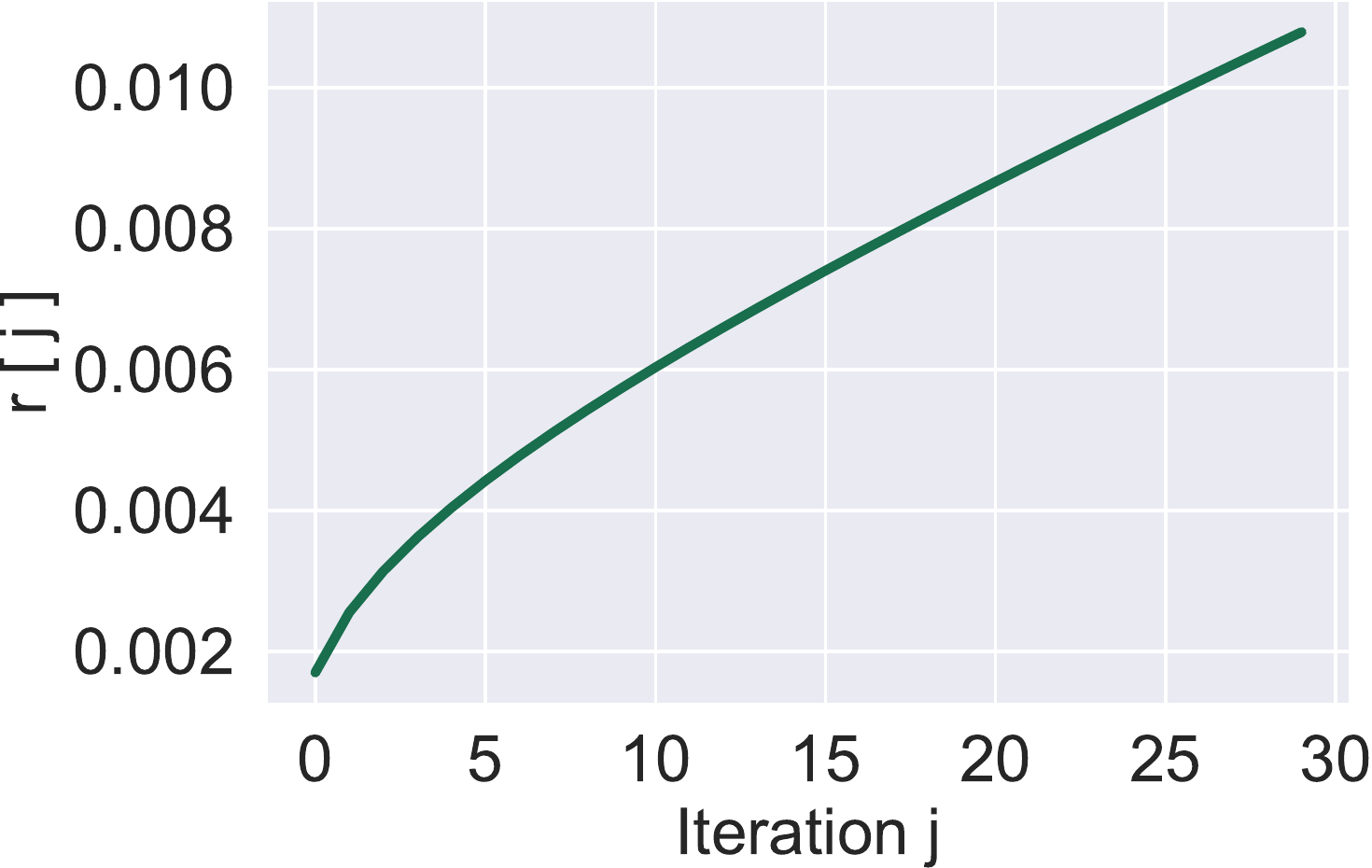}\label{fig:inf_radius}}
 \caption{Exemplary infection dynamics for $\alpha= 0.02$, $\beta=0.15$, $\gamma = 4.4 \times 10^{5}$, and $\delta=0.01$.}%
 \label{fig:infection}%
\end{figure}

\begin{figure*}[h!]%
 \centering
 \subfloat[$|\mathcal{I}_{\text{NTB}}|= 0.$
 ]{\includegraphics[width=0.18\linewidth]{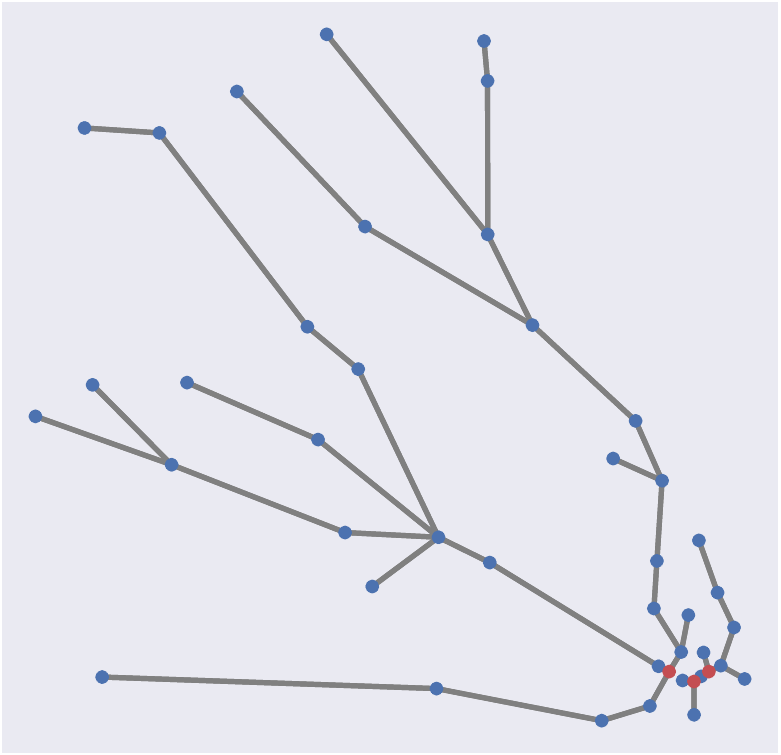}\label{fig:n0}} \hspace{2pt}
 \subfloat[$|\mathcal{I}_{\text{NTB}}|= 1.$
 ]{\includegraphics[width=0.18\linewidth]{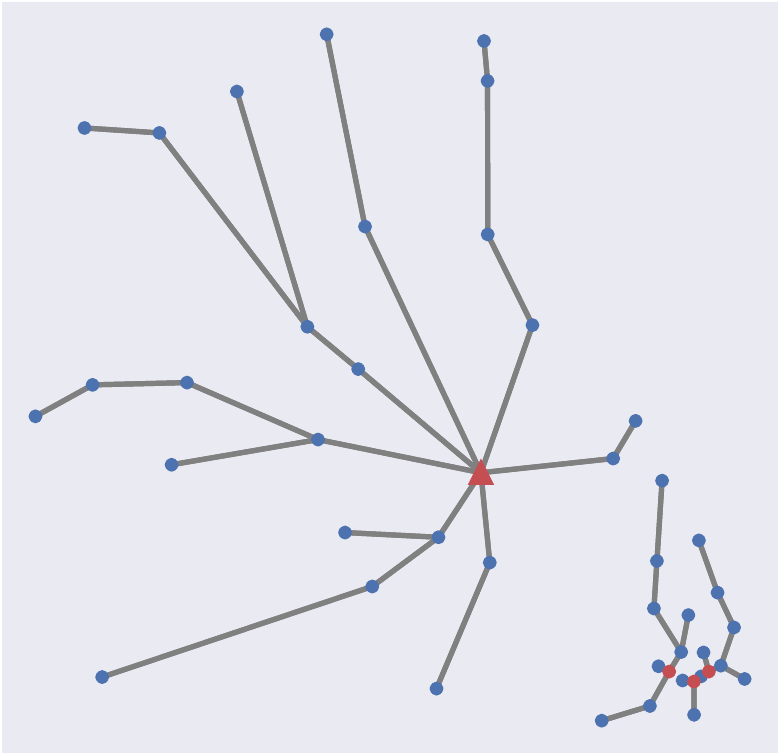}\label{fig:n1}} \hspace{2pt}
 \subfloat[$|\mathcal{I}_{\text{NTB}}|= 2.$
 ]{\includegraphics[width=0.18\linewidth]{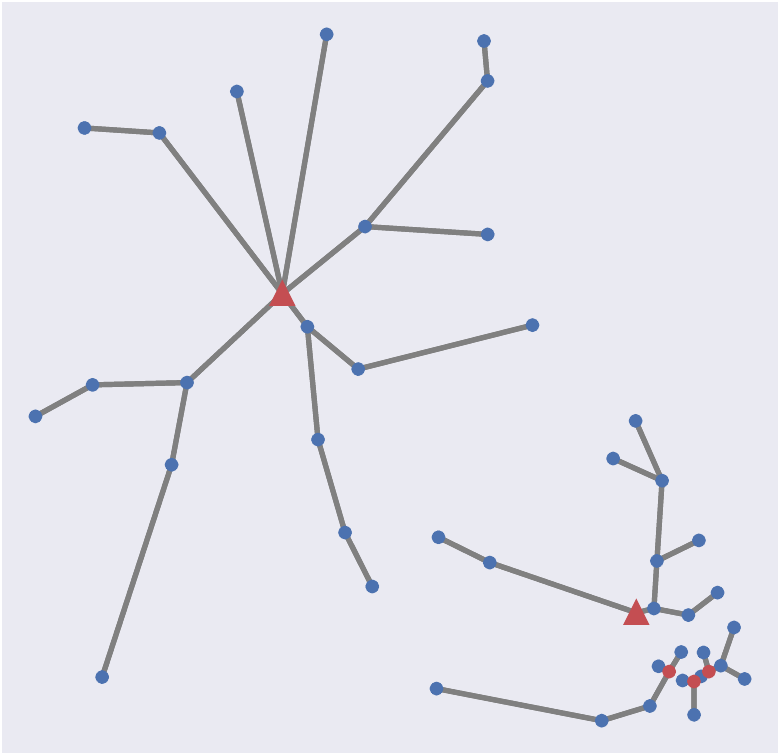}\label{fig:n2}} \hspace{2pt}
 \subfloat[$|\mathcal{I}_{\text{NTB}}|= 3.$
 ]{\includegraphics[width=0.18\linewidth]{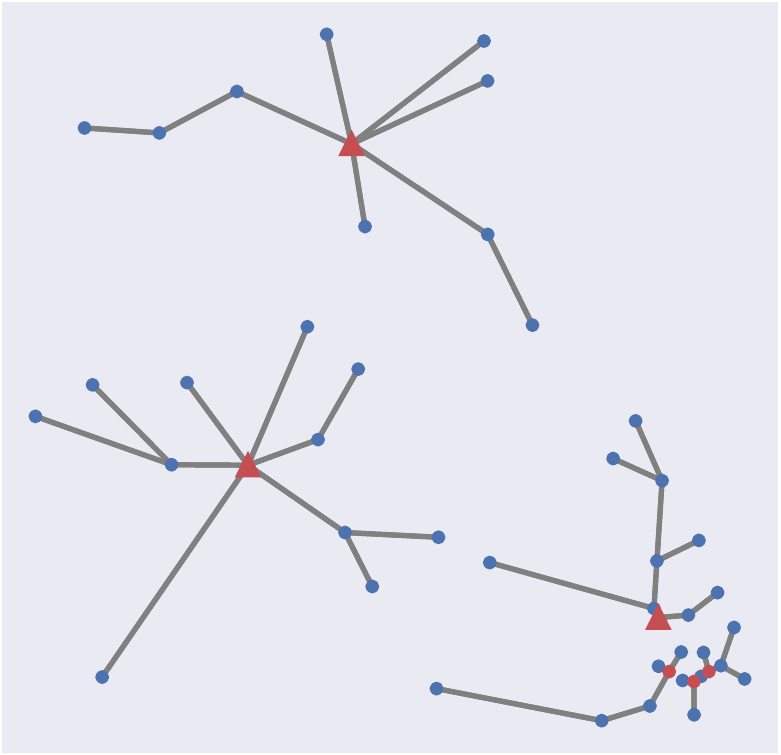}\label{fig:n3}} \hspace{2pt}
 \subfloat[$|\mathcal{I}_{\text{NTB}}|= 4.$
 ]{\includegraphics[width=0.18\linewidth]{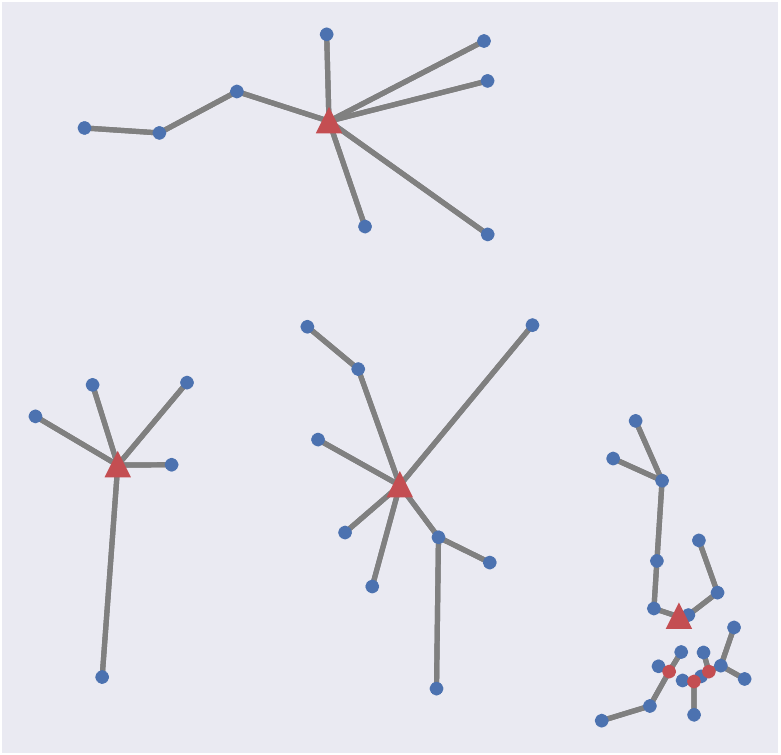}\label{fig:n4}} 
 \caption{Sample network generated as number of non-terrestrial BNs varies for $\alpha= 0.01$, $\beta=0.15$, $\gamma = 4.4\times10^{5}$, and $\delta=0.01$. The blue dots represent APs, red dots represent terrestrial BNs and red triangles represent non-terrestrial BNs. }%
 \label{fig:topology}%
\end{figure*}

The algorithm is presented in Algorithm~\ref{alg:zombiealgo}. The change in the infection speed and infection radius with iterations is shown in Fig.~\ref{fig:infection}. The infection radius increases quite rapidly in the beginning since the infected node wants to be faster than its infector in capturing the neighboring nodes.

\begin{algorithm}
\DontPrintSemicolon
\caption{Infection Algorithm}\label{alg:zombiealgo}
\KwInput{the population $\mathcal{U}$, the set of APs $\mathcal{A}^*$}
\KwOutput{the edges $\mathcal{E}^*$} 

$\mathcal{E}^* = \varnothing$ \Comment{Initially the  vertices are unconnected}\;
$j \gets 0$  \Comment{Initialize time-step to zero}\; 
$\mathcal{I}'_{j} = \mathcal{I}$ \Comment{Initialize the set of infected vertices} \;
\While{$|\mathcal{I}_j'| < |\mathcal{I}\cup \mathcal{A}^*|$}
{
    \For{$\mathbf{i} \in \mathcal{I}'_{j}$ }
    {
    Update $r_{\mathbf{i}}[j]$ according to \eqref{eq:radius} \;
        \For{$\mathbf{a} \in \mathcal{A}^* \setminus \mathcal{I}'_{j}$}
        {
            \If{$\lVert \mathbf{i} - \mathbf{a} \rVert < r_{\mathbf{i}}[j]$}
            {
                \If{$\mathbf{a} \notin \mathcal{I}'_{j+1}$}
                {
                    $\mathcal{E}^* \gets \mathcal{E}^* \cup \{ (\mathbf{i}, \mathbf{a}) \}$ \Comment{Generate edge} \;
                    $\mathcal{I}'_{j+1} = \mathcal{I}'_{j} \cup \{ \mathbf{a}\}$ \Comment{Infected} \;
                }
            }
        }
    Update $s_{\mathbf{i}}[j]$ according to \eqref{eq:speed}\;
    } 
    $j \gets j + 1$ \Comment{Increment time-step}\; 
}
\end{algorithm}


%


\subsection{Adding Non-Terrestrial Backhaul Nodes} \label{sec:non-terrestrial}
Relying entirely on the fixed terrestrial BNs may not be optimal. These terrestrial BNs are typically deployed in densely populated area, as is the case for the region in Kilimambogo, Kenya. Therefore, the APs in the sparse locations are connected indirectly to the BNs through many hops. To further improve the network configuration, we suggest the deployment of non-terrestrial BNs which receive backhaul from satellites or high altitude platforms. The position of these non-terrestrial BNs can be dynamically changed, which increases the chance to reliably connect the remote APs. We generate the graph using the infection algorithm by adding the non-terrestrial BNs to the initially infected vertices. The resulting graph for the Kilimambogo region as the number of non-terrestrial BNs increases is shown in Fig.~\ref{fig:topology}. 

With Network X library~\cite{networkx}, we perform analysis on the resulting networks and show the improvement in the network design in Fig.~\ref{fig:netx}. 
In Fig.~\ref{fig:netx}(a), we see that the average hop count decreases as we introduce more non-terrestrial BNs. Similarly, in Fig.~\ref{fig:netx}(b), the number of APs supported by each BN decreases. Adding more non-terrestrial BNs also makes the AP distribution per BN more fair, as seen in Fig.~\ref{fig:netx}(c). Moreover, this also lowers the use of backhaul links, and we see the reduction in the total backhaul length in Fig.~\ref{fig:netx}(d).
These results suggest that the addition of non-terrestrial BNs significantly improve the deployment of a realistic network.

\begin{figure}[h!]%
 \centering
 \subfloat[Average hop count.
 ]{\includegraphics[width=0.47\linewidth]{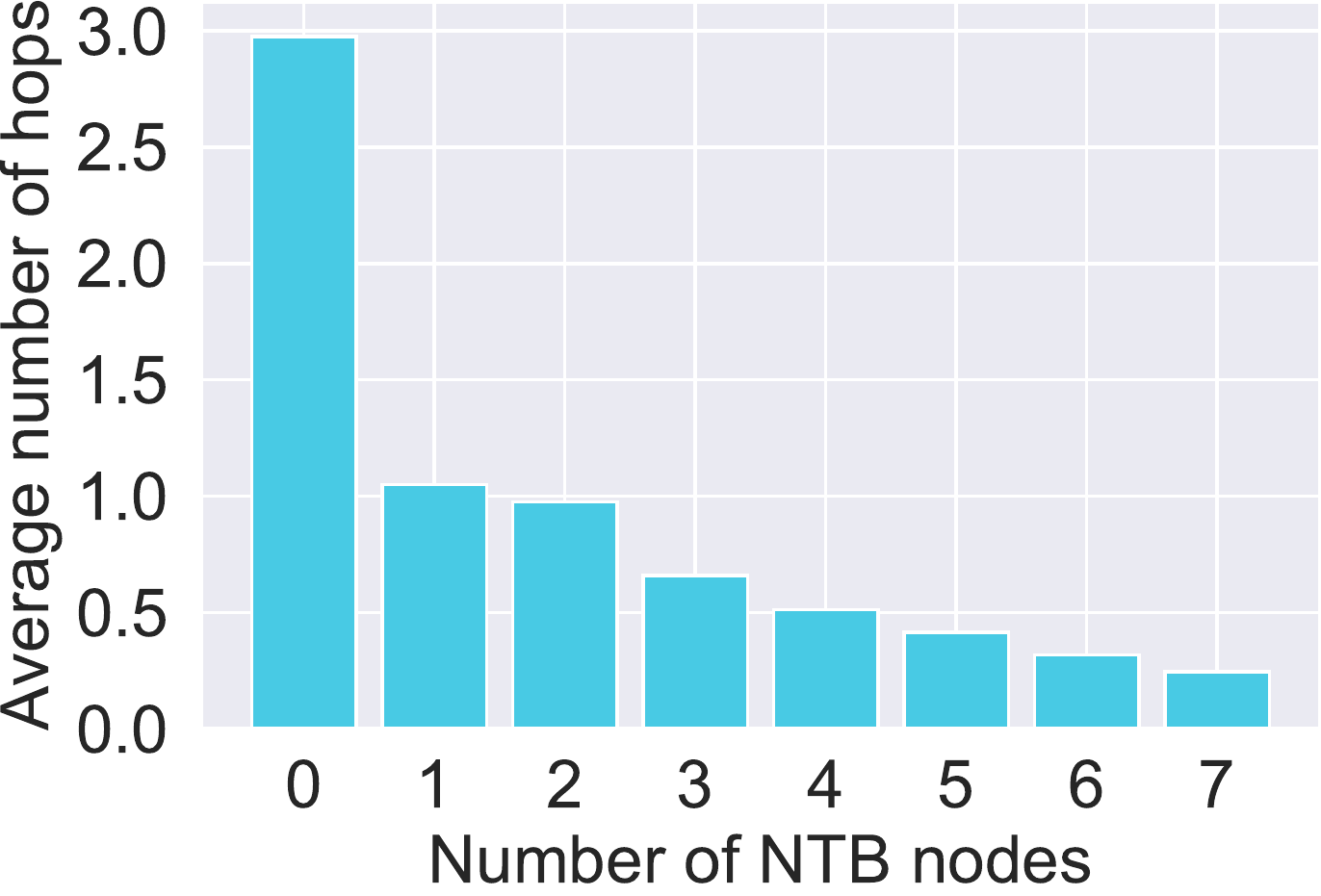}\label{fig:hops}} \hspace{2pt}
 \subfloat[AP count per BN.
 ]{\includegraphics[width=0.47\linewidth]{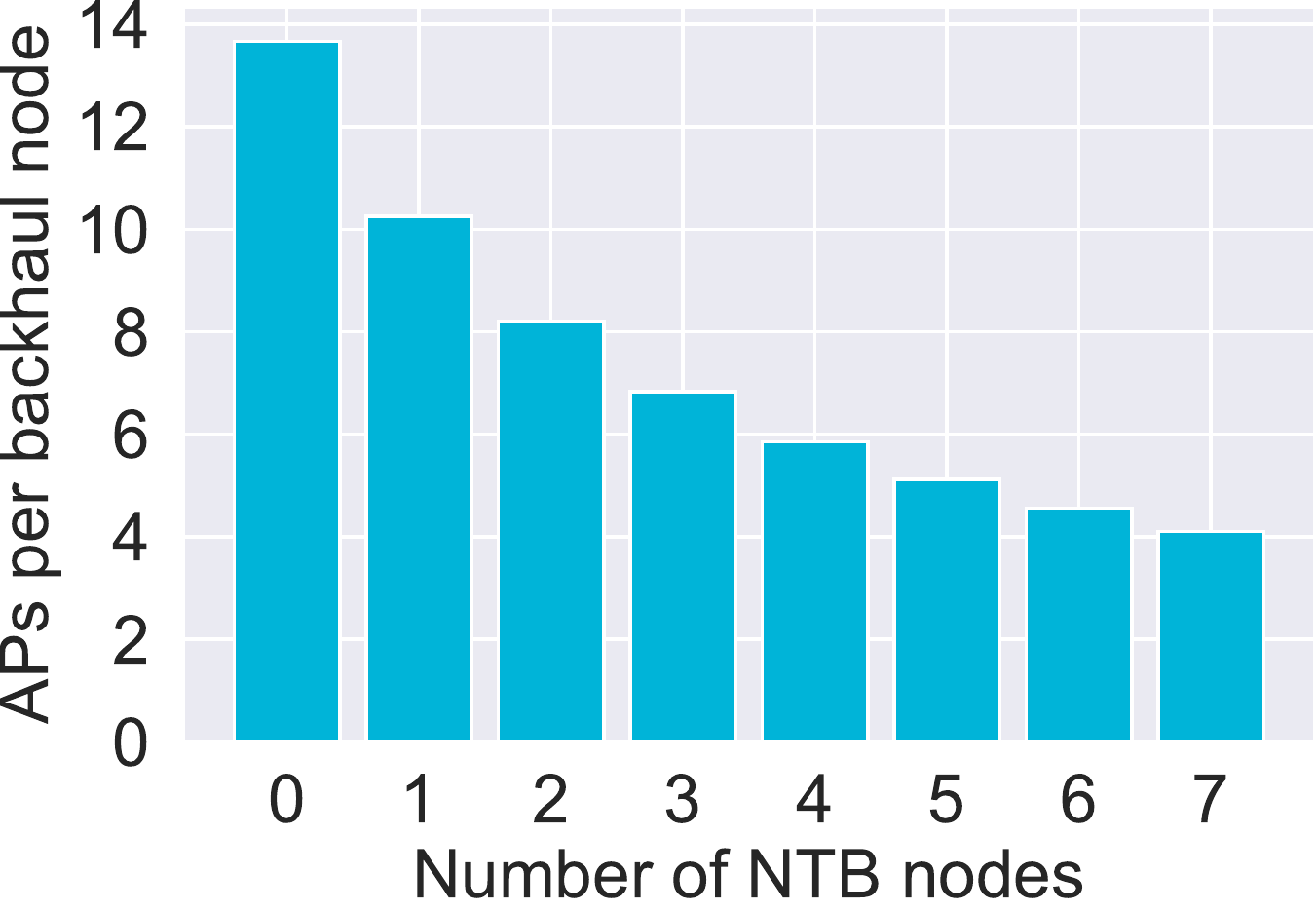}\label{fig:load}}\\
 \subfloat[Fairness in AP distribution.
 ]{\includegraphics[width=0.47\linewidth]{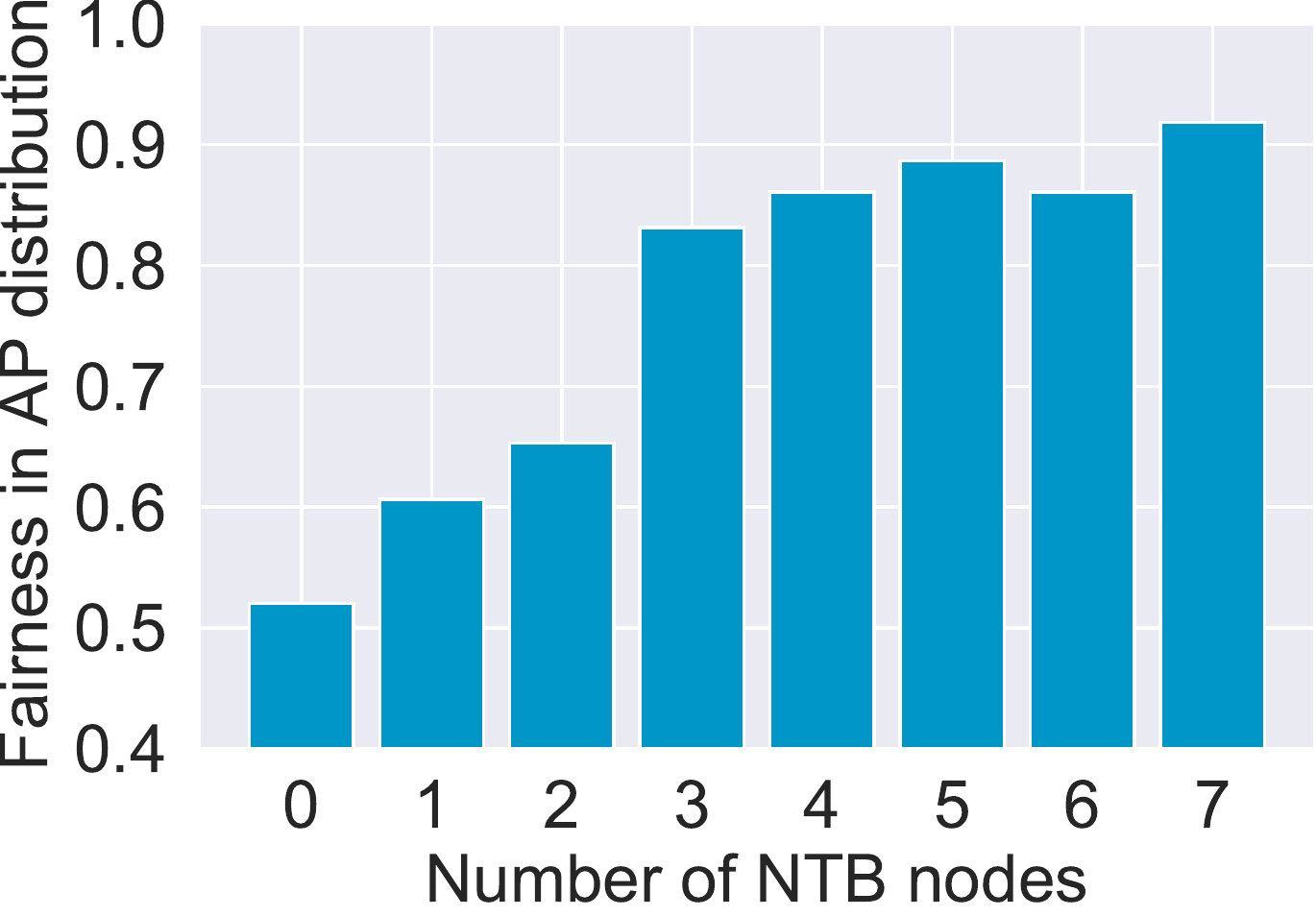}\label{fig:fair}} \hspace{2pt}
 \subfloat[Reduction in backhaul length.
 ]{\includegraphics[width=0.47\linewidth]{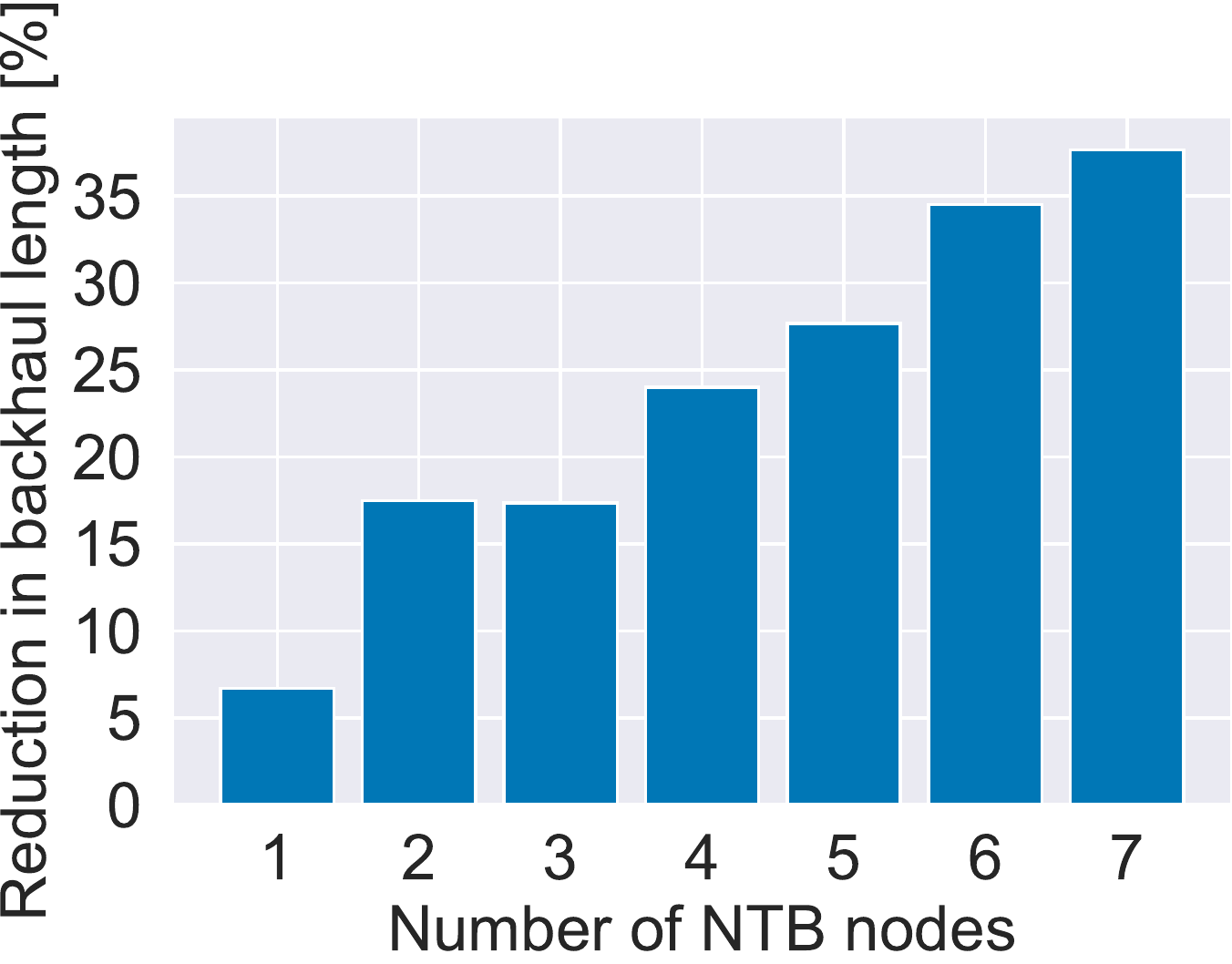}\label{fig:length}}
 \caption{Variation in key network parameters as the number of non-terrestrial backhaul nodes increases.}%
 \label{fig:netx}%
\end{figure}

\section{Conclusion} \label{sec:conclusion}
In this work, we have proposed an algorithmic pipeline to deploy a communication network to connect the unconnected population in remote/rural areas. We made use of the high resolution population data to first plan the AP deployment, and then connected them to the BNs in a cost-effective manner. To support the existing terrestrial BNs, we suggest the deployment of non-terrestrial BNs to further improve the network performance in terms of average hop count, traffic distribution, and backhaul length.
The next task is to choose the best backhaul type \cite{su2013comparative} based on the distance and geographical conditions. The number of backhaul and fronthaul nodes are constrained by the budget, but they also need to be sufficiently high to meet the traffic demands of the users. Solving this optimization problem relies on the subjective costs such as per capita GDP, and operational and capital expenses related to the infrastructure \cite{wu2007heterogeneous, chen2015financial, smail2017techno}.

\section*{Acknowledgement}
The work is partially supported by the Klaus Tschira Foundation through Alumnode Project Funding 2021-2022.

\appendix

\subsection{Covering and Packing} \label{app:covering_packing}
We define covering and packing in a two-dimensional space. See~\cite[Sec.~4.2]{vershynin2018high} for a reference.
\begin{definition}[Covering] \label{def:covering}
    An $\epsilon$-cover of a set $\mathcal{T}$ in $\mathbb{R}^2$ is a set $\{t_1,\dots,t_N\} \subset \mathcal{T}$ such that for all $t \in \mathcal{T}$ there exists an $i\in [\![ 1,N] \!]$ such that $\|t_i - t\| \le \epsilon$. The $\epsilon$-covering number $N(\epsilon,\mathcal{T})$ is the cardinality of the smallest $\epsilon$-cover.
\end{definition}

\begin{definition}[Packing] \label{def:packing}
    An $\epsilon$-packing of a set in $\mathbb{R}^2$ is a set $\{t_1,\dots,t_P\} \subset \mathcal{T}$ such that $\|t_i-t_j\| > \epsilon$ for all $i,j \in [\![ 1,P]\!]$. The $\epsilon$-packing number $P(\epsilon,\mathcal{T})$ is the cardinality of the largest $\epsilon$-packing.
\end{definition}

\begin{proposition} \label{prop:covering_number}
    It holds that
    \begin{align} \label{eq:covering_bound_1}
        P(2\epsilon,\mathcal{T}) \le N(\epsilon,\mathcal{T}) \le P(\epsilon,\mathcal{T}),
    \end{align}
    and 
    \begin{align} \label{eq:covering_bound_2}
        \frac{{\rm Area}(\mathcal{T})}{\pi \epsilon^2} \le N(\epsilon,\mathcal{T}) \le \frac{4{\rm Area}(\mathcal{T}_{\epsilon/2})}{\pi \epsilon^2},
    \end{align}
    where $\mathcal{T}_{\epsilon/2}$ denotes the union of the circles of radius $\epsilon/2$, each centered at a point in $\mathcal{T}$. ($\mathcal{T}_{\epsilon/2}$ is an inflated set of $\mathcal{T}$.)
\end{proposition}
\begin{proof}
    The bounds~\eqref{eq:covering_bound_1} and~\eqref{eq:covering_bound_2} follows Lemma~4.2.8 and Proposition~4.2.12, respectively, in~\cite{vershynin2018high}. 
\end{proof}

\bibliographystyle{IEEEtran}
\bibliography{main.bib}

\begin{thebibliography}{10}
\providecommand{\url}[1]{#1}
\csname url@samestyle\endcsname
\providecommand{\newblock}{\relax}
\providecommand{\bibinfo}[2]{#2}
\providecommand{\BIBentrySTDinterwordspacing}{\spaceskip=0pt\relax}
\providecommand{\BIBentryALTinterwordstretchfactor}{4}
\providecommand{\BIBentryALTinterwordspacing}{\spaceskip=\fontdimen2\font plus
\BIBentryALTinterwordstretchfactor\fontdimen3\font minus
  \fontdimen4\font\relax}
\providecommand{\BIBforeignlanguage}[2]{{%
\expandafter\ifx\csname l@#1\endcsname\relax
\typeout{** WARNING: IEEEtran.bst: No hyphenation pattern has been}%
\typeout{** loaded for the language `#1'. Using the pattern for}%
\typeout{** the default language instead.}%
\else
\language=\csname l@#1\endcsname
\fi
#2}}
\providecommand{\BIBdecl}{\relax}
\BIBdecl

\bibitem{SustainableDevelopment2021}
\BIBentryALTinterwordspacing
{The ITU/UNESCO Broadband Commission for Sustainable Development}, ``The state
  of broadband 2021: People-centred approachesfor universal broadband,'' Sep.
  2021. [Online]. Available:
  \url{https://itu.int/itu-d/reports/broadbandcommission/state-of-broadband-2021/}
\BIBentrySTDinterwordspacing

\bibitem{chaoub20216g}
A.~Chaoub, M.~Giordani, B.~Lall, V.~Bhatia, A.~Kliks, L.~Mendes, K.~Rabie,
  H.~Saarnisaari, A.~Singhal, N.~Zhang \emph{et~al.}, ``{6G} for bridging the
  digital divide: Wireless connectivity to remote areas,'' \emph{IEEE Wireless
  Communications}, 2021.

\bibitem{dang2021big}
S.~Dang, C.~Zhang, B.~Shihada, and M.-S. Alouini, ``Big communications: Connect
  the unconnected,'' \emph{arXiv preprint arXiv:2104.06131}, 2021.

\bibitem{yaacoub2020efficient}
E.~Yaacoub and M.-S. Alouini, ``Efficient fronthaul and backhaul connectivity
  for {IoT} traffic in rural areas,'' \emph{IEEE Internet of Things Magazine},
  vol.~4, no.~1, pp. 60--66, 2020.

\bibitem{Hamid2011self}
K.~Ab-Hamid, C.~E. Tan, and S.~P. Lau, ``Self-sustainable energy efficient long
  range wifi network for rural communities,'' in \emph{IEEE GLOBECOM
  Workshops}, Houston, TX, USA, Dec. 2011, pp. 1050--1055.

\bibitem{Maurilio2021coverage}
M.~Matracia, M.~A. Kishk, and M.-S. Alouini, ``Coverage analysis for
  {UAV}-assisted cellular networks in rural areas,'' \emph{IEEE Open Journal of
  Vehicular Technology}, vol.~2, pp. 194--206, Apr. 2021.

\bibitem{ogutu2021techno}
O.~B. Ogutu and E.~J. Oughton, ``A techno-economic cost framework for satellite
  networks applied to low earth orbit constellations: Assessing starlink,
  oneweb and kuiper,'' \emph{arXiv preprint arXiv:2108.10834}, 2021.

\bibitem{viasat2021}
\BIBentryALTinterwordspacing
Viasat. (2021) Viasat global {Ka}-band coverage. [Online]. Available:
  \url{https://www.viasat.com/space-innovation/satellite-fleet/global-satellite-internet/}
\BIBentrySTDinterwordspacing

\bibitem{khaturia2020connecting}
M.~Khaturia, P.~Jha, and A.~Karandikar, ``Connecting the unconnected: Toward
  frugal 5g network architecture and standardization,'' \emph{IEEE
  Communications Standards Magazine}, vol.~4, no.~2, pp. 64--71, 2020.

\bibitem{inaba1994applications}
M.~Inaba, N.~Katoh, and H.~Imai, ``Applications of weighted voronoi diagrams
  and randomization to variance-based k-clustering,'' in \emph{Proceedings of
  the tenth annual symposium on Computational geometry}, 1994, pp. 332--339.

\bibitem{arthur2006slow}
D.~Arthur and S.~Vassilvitskii, ``How slow is the k-means method?'' in
  \emph{Proceedings of the twenty-second annual symposium on Computational
  geometry}, 2006, pp. 144--153.

\bibitem{scikit-learn}
F.~Pedregosa, G.~Varoquaux, A.~Gramfort, V.~Michel, B.~Thirion, O.~Grisel,
  M.~Blondel, P.~Prettenhofer, R.~Weiss, V.~Dubourg, J.~Vanderplas, A.~Passos,
  D.~Cournapeau, M.~Brucher, M.~Perrot, and E.~Duchesnay, ``Scikit-learn:
  Machine learning in {P}ython,'' \emph{Journal of Machine Learning Research},
  vol.~12, pp. 2825--2830, 2011.

\bibitem{facebook2021}
\BIBentryALTinterwordspacing
Facebook. (2021, Sep.) High resolution population density maps. {Humanitarian
  Data Exchange (HDX)}. [Online]. Available:
  \url{https://dataforgood.facebook.com/dfg/tools/high-resolution-population-density-maps}
\BIBentrySTDinterwordspacing

\bibitem{opencellid2021}
\BIBentryALTinterwordspacing
{Open Cell ID}. (2021, Oct.) Cell tower data. License: CC-BY-SA 4.0. [Online].
  Available: \url{https://my.opencellid.org/}
\BIBentrySTDinterwordspacing

\bibitem{liu2003propagation}
Z.~Liu, Y.-C. Lai, and N.~Ye, ``Propagation and immunization of infection on
  general networks with both homogeneous and heterogeneous components,''
  \emph{Physical Review E}, vol.~67, no.~3, p. 031911, 2003.

\bibitem{may2001infection}
R.~M. May and A.~L. Lloyd, ``Infection dynamics on scale-free networks,''
  \emph{Physical Review E}, vol.~64, no.~6, p. 066112, 2001.

\bibitem{networkx}
A.~Hagberg, P.~Swart, and D.~S~Chult, ``Exploring network structure, dynamics,
  and function using {NetworkX},'' Los Alamos National Lab.(LANL), Los Alamos,
  NM (United States), Tech. Rep., 2008.

\bibitem{su2013comparative}
X.~Su and K.~Chang, ``A comparative study on wireless backhaul solutions for
  beyond 4g network,'' in \emph{The International Conference on Information
  Networking 2013 (ICOIN)}.\hskip 1em plus 0.5em minus 0.4em\relax IEEE, 2013,
  pp. 505--510.

\bibitem{wu2007heterogeneous}
C.-H. Wu and Y.-C. Chung, ``Heterogeneous wireless sensor network deployment
  and topology control based on irregular sensor model,'' in
  \emph{International Conference on Grid and Pervasive Computing}.\hskip 1em
  plus 0.5em minus 0.4em\relax Springer, 2007, pp. 78--88.

\bibitem{chen2015financial}
Y.~Chen, L.~Duan, and Q.~Zhang, ``Financial analysis of 4g network
  deployment,'' in \emph{2015 IEEE Conference on Computer Communications
  (INFOCOM)}.\hskip 1em plus 0.5em minus 0.4em\relax IEEE, 2015, pp.
  1607--1615.

\bibitem{smail2017techno}
G.~Smail and J.~Weijia, ``Techno-economic analysis and prediction for the
  deployment of 5g mobile network,'' in \emph{2017 20th Conference on
  innovations in clouds, internet and networks (ICIN)}.\hskip 1em plus 0.5em
  minus 0.4em\relax IEEE, 2017, pp. 9--16.

\bibitem{vershynin2018high}
R.~Vershynin, \emph{High-dimensional probability: {A}n introduction with
  applications in data science}.\hskip 1em plus 0.5em minus 0.4em\relax
  Cambridge University Press, 2018, vol.~47.

\end{thebibliography}

\end{document}